\documentclass[twocolumn,showpacs,floatfix,nofootinbib,showkeys]{revtex4-1}
\usepackage{amsmath, amssymb, amsfonts}
\usepackage{graphicx}
\usepackage{hyperref}

\usepackage{bbold}

\usepackage{slashed}
\usepackage[T1]{fontenc}
\usepackage[utf8]{inputenc}
\usepackage[english]{babel}
\usepackage{amsmath} 
\usepackage{latexsym}
\usepackage{color}
\usepackage{bm}
\usepackage{amsthm}
\usepackage[nottoc]{tocbibind}
\usepackage{amsmath}
\usepackage{amssymb}
\usepackage{graphicx}
\usepackage{hyperref}
\usepackage{comment}

\usepackage{texnames}

\hypersetup{
linktocpage=true,
colorlinks=true,
linkcolor=blue,          
citecolor=mygreen,        
filecolor=blue,      
urlcolor=blue           
}  
\usepackage{color}
\definecolor{mygreen}{rgb}{0.0,0.75,0.0}\usepackage{timestamp}

\usepackage{etoolbox}
\appto\UrlBreaks{\do\-}
\PassOptionsToPackage{hyphens}{url}
\usepackage[hyphenbreaks]{breakurl}

\newtheorem{theorem}{Theorem}[section]

\newtheorem{prop}{Proposition}[section]

\newtheorem{cor}{Corollary}
\newtheorem{defi}{Definition}

\usepackage{fancyhdr}
\def\beq{\begin{equation}}
\def\eeq{\end{equation}}
\def\beqa{\begin{eqnarray}}
\def\eeqa{\end{eqnarray}}
\def\ben{\begin{enumerate}}
\def\een{\end{enumerate}}

\def\bit{\begin{itemize}}
\def\eit{\end{itemize}}

\newcommand{\ket}[1]{|{#1}\rangle}
\newcommand{\N}[1]{\widetilde \nu}
\newcommand{\Nm}[1]{\ket{\widetilde \nu^{(m)}_{#1}}}
\newcommand{\Nf}[1]{\ket{\widetilde\nu^{(f)}_{#1}}}
\newcommand{\nuf}[1]{\ket{\nu^{(f)}_{#1}}}
\newcommand{\nuem}[1]{\ket{\nu^{(m)}_{#1}}}

  \def \upmns {{\rm{U_{PMNS}}}}

\begin{document}

\title{
General neutrino mass spectrum and mixing properties in seesaw mechanisms 
}
\author{Wojciech Flieger$^1$}
\author{Janusz Gluza$^{1,2}$}
\affiliation{
\centerline{$^1$Institute of Physics, University of Silesia, Katowice, Poland 
}
$^2$Faculty of Science, University of Hradec Kr\'alov\'e, Czech Republic
}
\date{\today}

\medskip

\numberwithin{equation}{section}

\begin{abstract}

Neutrinos stand out among elementary particles through their unusually small masses.
Various seesaw mechanisms attempt to explain this fact. 
In this work applying insights from matrix theory we are in a position to treat variants of seesaw mechanisms in a general manner.
Specifically, using Weyl's inequalities we discuss and rigorously prove under which conditions the seesaw framework leads to a mass spectrum with exactly three light neutrinos. 
We find an estimate on the mass of heavy neutrinos to be the mass obtained by neglecting light neutrinos shifted at most by the maximal strength of the coupling to the light neutrino sector.
We provide analytical conditions allowing to prescribe that precisely two out of five neutrinos are heavy.  
For higher-dimensional cases the inverse eigenvalue methods are used. 
In particular, for the CP invariant scenarios we show that if the  neutrino sector has a valid mass matrix after neglecting the light ones, i.e. the respective mass submatrix is positive definite, then large masses are provided by matrices with large elements accumulated on the diagonal. Finally, the Davis-Kahan theorem is used to show how masses affect the rotation of light neutrino eigenvectors from the standard Euclidean basis. 
This general observation concerning neutrino mixing together with results on the mass spectrum properties opens directions for further neutrino physics studies using matrix analysis. 
\end{abstract}


\maketitle
\allowdisplaybreaks

\section{Introduction}

The Standard Model (SM) of electroweak interactions is based on the $SU(2)_L \otimes U(1)_Y$ gauge group \cite{PhysRevLett.19.1264,Glashow:1961tr,Salam:1968rm} which determines the set of the gauge boson fields.
On the other hand the gauge group alone does not imply uniquely what kind and range of elementary particles can exist in nature \cite{Weinberg:1995mt}. 
The set of matter fields presently considered to be the elementary particles is  based on a great number of experimental insights which were the result of a long-standing research program. 
It should hence be noted that the experimental observations are the deciding factor in choosing the matter content that makes up the theory of elementary particles. 
Any hypothetical signals that could not be explained by the SM, like lepton violating processes, would need modification of the matter content and interactions. That choice must be based on experimental evidence. 

As far as neutrinos are concern, which are the main theme of this work, presently three neutrinos are known of different flavours, to which correspond three charged leptons. That the three light neutrino species exist has been known since the LEP time. The central value for the effective number of light neutrinos $N_\nu$ was determined by analyzing around 20 million  $Z$-boson decays, yielding $N_\nu=2.9840 \pm 0.0082$ \cite{ALEPH:2005ab,Novikov:1999af}. It is worth mentioning that the recent reevaluation of the data \cite{Voutsinas:2019hwu,Janot:2019oyi}, including higher order QED corrections to the Bhabha process, constrain further the value of $N_\nu$, which is now $N_\nu= 2.9963 \pm 0.0074$.  The new $N_\nu$ value is much closer to 3. Including shrunk of the error it leaves less space for non-standard neutrino mixings. 
In fact, a natural extension of the SM by right-handed neutrinos leads to a theoretical prediction with $N_\nu$ less than three \cite{Jarlskog:1990kt}, assuming that there are non-zero mixings of active and sterile neutrinos, which implies non-unitarity of the matrix responsible for mixings among three known neutrino states. 
This can be seen from the general neutrino mixing setting. Let us denote a three dimensional space which describes known neutrino mass and flavor states by ${ \nuem \alpha }$ and 
${ \nuf \alpha }$, respectively. Any extra, beyond SM (BSM) mass and flavor states we denote by $\Nm j$ and $\Nf j$ for $j=1,\ldots, n_R$, respectively.
In this general scenario mixing between an extended set of neutrino mass states $\{\nuem \alpha,\Nm \beta\}$ with flavor states $\{\nuf \alpha,\Nf \beta\}$ is described by
\begin{align}
\begin{pmatrix} { \nuf \alpha} \\ 
 \Nf \beta\end{pmatrix}  &=
 \begin{pmatrix} {{ \upmns }} & V_{lh} \\  V_{hl} & V_{hh}
  \end{pmatrix} 
\begin{pmatrix} {\nuem \alpha } \\ 
\Nm \beta\end{pmatrix} 
\equiv U  \begin{pmatrix} {{ \nuem \alpha }} \\ \Nm \beta\end{pmatrix}\;.
\label{ugen}
\end{align}

The observable part of the above is the transformation from mass $\nuem \alpha, \Nm \beta$ to SM flavor $\nuf \alpha$ states and reads 
\begin{equation}
\nuf \alpha = 
\sum_{i=1}^3\underbrace{\left( \upmns  \right)_{\alpha i} {\nuem i}}_{\rm SM \;part}
+ 
\sum_{j=1}^{n_R}\underbrace{
{ \left(V_{lh}  \right)}_{\alpha j} \Nm j}
_{\rm BSM \; part}\;. \label{gen}
\end{equation}  

If $\upmns$ is not unitary 
then  there necessarily is a light-heavy neutrino "coupling" and the mixing between sectors is nontrivial $V_{lh} \neq 0 \neq V_{hl}$. 
As $U$ in (\ref{ugen}) is unitary, and we know from neutrino oscillation experiments the $\upmns{}$ matrix\footnote{Acronym comes from  authors: Pontecorvo, Maki, Nakagawa, Sakata who introduced idea of oscillations to neutrino physics \cite{Pontecorvo:1957qd,Maki:1962mu}.}
is with experimental accuracy unitary, it follows that elements of the non-diagonal matrices $V_{lh},  V_{hl}$ in (\ref{ugen}) which are responsible for mixings of known neutrinos with extra states must be very small. 

There is a natural explanation of the above structure of $\upmns$ and $V_{lh},V_{hl}$ matrices, and it comes with the celebrated seesaw mechanism which in first place explains small masses of known neutrinos. This mechanism justifies also introduction of indices "$l$" and "$h$" in \eqref{ugen} which stand for "light" and "heavy" as usually we expect extra neutrino species to be much heavier than known neutrinos.

To get physical masses, in the seesaw mechanism, the unitary matrix $U$ (\ref{ugen}) is used to diagonalize the general neutrino mass matrix

\begin{equation}\label{ssmm}
M_{SS}=
\left(
\begin{array}{cc}
M_L & M_{D} \\
M_{D}^T & M_{R} 
\end{array}
\right) ,
\end{equation}
\noindent
using a congruence transformation
\begin{equation}
U^{T}M_{SS}U \simeq diag(M_{light}, M_{heavy}).
\label{congr}
\end{equation}

The exact form and origin of the neutrino mass matrices $M_L, M_D$ and $M_R$ are not relevant now, they will be specified in the next section. In general, with the assumption that $M_L=0$, and 
\begin{equation}
\vert M_{D} \vert \ll \vert M_{R} \vert ,
\label{restss}
\end{equation}  
which means that elements of $M_D$ are much smaller than elements of $M_R$ with respect to absolute values, where $\vert \cdot \vert$ in case of matrices denotes absolute values of elements,

\begin{eqnarray}
M_{light} &\simeq &-M_{D}M_{R}^{-1}M_{D}^{T}, \label{seesaw1}\\
M_{heavy} &\simeq & M_{R}. \label{seesaw2}
\label{seesaw2}
\end{eqnarray}
A large scale of $M_R$, makes $M_{light}$ small, which represents the main idea of the seesaw mechanism.

This mechanism was proposed for the first time by Minkowski in 1977 \cite{Minkowski:1977sc}. It originates from the idea of Grand Unified Theory in which heavy neutrino mass states are present. Such neutrinos are supposed to be sterile, i.e. they are insensitive with respect to the weak interaction. In  \cite{Minkowski:1977sc} the model based on $SU(2)_{L} \times SU(2)_{R} \times U(1)$ gauge symmetry was considered with its consequences for $\mu \rightarrow e\gamma$ decays. At that time only 2 fermion families were under consideration, but both seesaw mass matrix and famous seesaw formulas, which will be presented later, were introduced. Afterwards, similar models of neutrino mass generation were discussed in 1979 \cite{Yanagida:1979as, GellMann:1980vs} and 1980 \cite{Glashow:1979nm, Barbieri:1979ag, Mohapatra:1979ia}. Authors observed that smallness of the mass of neutrinos can be explained when super heavy right-handed Majorana neutrinos are introduced, which were considered as the result of grand unification models such as $SO(10)$ theories or as the consequence of the horizontal symmetry. In both cases, the small neutrino mass appears as a consequence of the symmetry breaking. 

Nowadays there is a plethora of seesaw models.  
Ranging neutrino masses from zero to the GUT scale, mass mechanisms introduce different neutrino states \cite{Drewes:2013gca}. Apart from Dirac or Majorana types, there are pseudo-Dirac (or quasi-Dirac) \cite{Akhmedov:2014kxa}, schizophrenic \cite{Allahverdi:2010us}, or vanilla  \cite{Dev:2013vba} neutrinos, to call some of them.  
Popular seesaw mechanisms give a possibility for a dynamical explanation why known active neutrino states are so light. They appear to be of  Majorana type (recently dynamical explanation for Dirac light neutrinos has been proposed \cite{Valle:2016kyz}).
 Seesaw type-I models have been worked out in  \cite{Minkowski:1977sc,GellMann:1980vs,Yanagida:1980xy,Mohapatra:1980yp},
  type-II in \cite{Magg:1980ut},
  type-III in \cite{Foot:1988aq}. A hybrid mechanism is also possible
  \cite{Franco:2015pva}.
For inverse seesaw, see \cite{Mohapatra:1986aw,Mohapatra:1986bd}. Some of recent and interesting works on seesaw mechanisms which touch also cosmological and lepton flavor violation issues are \cite{Hernandez:2016kel,Ballesteros:2016euj,Drewes:2016jae,Okada:2017dqs,Das:2017nvm,Bhardwaj:2018lma,Branco:2019avf,Dror:2019syi,Dev:2019rxh,He:2020zns,Bolton:2019pcu}.
  
In a present work, we extend the approach defined in our previous works \cite{Bielas:2017lok,Flieger:2019eor}, where neutrino mixing matrices were considered from the point of view of matrix analysis, to the case of seesaw scenarios. In  \cite{Bielas:2017lok} we argued that singular values of mixing matrices and contractions applied to interval mixing matrices determine possible BSM effects in oscillation parameters and can be used to define the physical neutrino mixing space. 
In addition, a procedure of matrix dilation gives a possibility to find BSM extensions based on experimental data given for PMNS mixing matrices. In this way, we are closer to understand a long-standing puzzle in neutrino physics, namely if and what kind of extensions with extra neutrino states are possible, beyond the known three light neutrinos mixing picture. Using these techniques, new stringent limits for the light-heavy neutrino mixings in the  3+1 scenario (three active, light neutrinos plus one extra sterile neutrino state) have been obtained \cite{Flieger:2019eor}.  

What follows, we discuss a second part of the neutrino puzzle, focusing on the neutrino mass matrices and trying to figure out how much information the rigid structure of mass matrices characteristic for seesaw mechanisms provides about the neutrino mass spectrum.
In a similar spirit, a perturbation theory was used in \cite{Czakon:2001we} to prove that in the seesaw type I scenario we cannot get the fourth light neutrino. 
This proof was based on a standard seesaw assumption that elements of the heavy neutrino sector represented conventionally by Majorana mass matrix $M_{R}$ are much larger than elements of the Dirac mass matrix $M_{D}$, $|M_R|\gg|M_D|$. 
Besnard \cite{Besnard:2016tcs} gave an elegant proof, using the min-max theorem, that in the seesaw scenario there is a gap in the spectra. In the proof, the author assumed that the whole mass matrix is not singular and this way excluded a massless neutrino. However, current experimental data do not exclude the possibility of one massless neutrino. Also the assumption $|M_R| \gg |M_D|$ is not sufficient in general. It can be seen in the simplest way by considering $M_R$ matrix with all elements much larger than these of $M_D$, but all equal. In this case, taking, for instance, $M_R$ three dimensional, a rank of this matrix is 2, so one eigenvalue is zero. Only this example shows that a relation between matrix structures and derived eigenvalues is complicated.
Moreover, considering elements {\it of the same order} may be misleading and inaccurate. Even a simple matrix
\begin{equation}
A=
\left(
\begin{array}{cc}
100 & -95 \\
-95 & 90
\end{array}
\right)\label{eq:example_1}
\end{equation}
results in two completely different scales of eigenvalues $\lambda(A)= \lbrace 190.131, -0.131 \rbrace$. To infer eigenvalues and eigenvectors from the structured, large-dimensional matrices which pose different scales of elements is not trivial. We will examine also a connection between masses and mixings for the generic seesaw model.  

The structure of the paper is following. In the next section we will discuss different ways in which seesaw models can be realized.
In Section III  main results are obtained for the neutrino mass spectrum. In a scenario with two sterile neutrinos, analytic entrywise bounds for heavy neutrinos are presented. A higher dimensional situation is discussed using inverse eigenvalue methods for the positive definite matrices only. Also an alternative proof to \cite{Czakon:2001we} is given, showing that for the seesaw mass matrix with hierarchical block-structured submatrices there are only 3 light neutrino states. In addition, it is shown how large splits among heavy neutrino states can occur. In Section IV we discuss an angle between subspaces of the eigenvalues which connects masses and mixings.
In the last section, we conclude our work and present possible directions for further studies of the neutrino mass and mixing matrices.
The work is supported by Appendix where details on matrix theory needed for refining studies of the mass matrix structures are given. 

\section{Seesaw types of mass matrices}  

\subsection{Standard seesaw mechanisms}

Oscillation experiments established that neutrinos {\it are not massless} \cite{Fukuda:1998mi,Ahmad:2002jz}, and we already know that at least two of three known neutrinos are massive. It calls for introducing \emph{massive} right-handed neutrino states to the matter content of the theory. Then, similarly to the quark sector where right-handed quark fields are present, right-handed neutrino fields $\nu_R$ lead to the Dirac mass term in the mass Lagrangian 
\begin{equation}
    \mathcal{L}_{D}= - \bar{\nu}_{L}M_{D}\nu_{R} + h.c., 
\end{equation}
where $M_{D}$ is a complex $3 \times n_R $ matrix. 
Now, allowing for self-conjugating Majorana fields 
\begin{equation}
    \nu^{\mathcal{C}}=\mathcal{C}\bar{\nu}^{T}= \nu ,
    \label{numaj}
\end{equation}
which relates right-handed and left-handed fermionic fields, $\nu_{R}=(\nu_{L})^{\mathcal{C}},  \nu_{L}=(\nu_{R})^{\mathcal{C}}$,
another mass term 
\begin{equation}
    \mathcal{L}_{M}=-\frac{1}{2}\overline{\nu}_{L}M_{L}(\nu_{L})^{\mathcal{C}}+ h.c.
\end{equation}
can be constructed.  $M_{L}$ is a $3 \times 3$ complex symmetric matrix build exclusively from left-handed chiral fields, making a description more economic\footnote{In his seminal work \cite{Majorana:1937vz,majoranatransl}, Majorana wrote "Even though it is perhaps not yet possible to ask experiments to decide between the new theory
and a simple extension of the Dirac equations to neutral
particles, one should keep in mind that the new theory
introduces a smaller number of hypothetical entities, in
this yet unexplored field."}.

In the same way $M_R$  can be constructed with $n_R$ right-handed chiral fermionic fields. In general, the mass Lagrangian can include both Dirac and Majorana terms 
\begin{equation}
    \mathcal{L}_{D+M}= - \bar{\nu}_{L}M_{D}\nu_{R} -\frac{1}{2} \bar{\nu}_{L}M_{L}(\nu_{L})^{\mathcal{C}} 
-\frac{1}{2} \overline{(\nu_{R})^{\mathcal{C}}}M_{R}\nu_{R} + h.c. .
\end{equation}
In the Dirac-Majorana mass term resulting fields are of Majorana type. It is possible to write it in a compact form, in which it resembles the Majorana mass term 
\begin{equation}
\mathcal{L}_{D+M}=-\frac{1}{2}\bar{n}_{L}M_{D+M}n_{L}^{\mathcal{C}} + h.c., \label{eqdm}
\end{equation}
where $n_{L}=\left( \nu_{L},(\nu_{R})^{\mathcal{C}}\right)^{T}$ fullfils (\ref{numaj}).

The Dirac-Majorana mass term (\ref{eqdm}) underlies a seesaw mechanism of the neutrino mass generation which tries to explain small masses of know neutrinos by assuming large masses of sterile neutrinos. 

Here we discuss in more details what has been mentioned in Introduction.
First, we assume that left-handed Majorana mass term $\mathcal{L}_{L}$ vanishes, since it is forbidden by SM symmetries, or it may result from higher-dimensional operators \cite{Weinberg:1979sa}, which effectively damp the order of magnitude of $M_L$ elements below that of $M_D$. Secondly, we assume that Dirac neutrinos acquire masses through standard Higgs mechanism, so the elements of the Dirac neutrino mass matrix are of the order of the electro-weak scale. Last, but maybe the most important assumption concerns right-handed Majorana mass term, which is a manifestation of new physics beyond SM, and tells us that right-handed neutrinos are very heavy particles. Altogether, we assume: 
\begin{equation}
M_{L} \simeq 0, \;M_{D} \sim {\rm EW-}{\rm scale}  \ll M_R \sim {\rm GUT-}{\rm scale}
\end{equation}
where ${\rm EW-}{\rm scale}$ refers to the electroweak spontaneous symmetry scale ($10^2$ GeV) and the ${\rm GUT-}{\rm scale}$ was originally taken to be of the order of $10^{15}$ GeV.

Hence, we get the $2 \times 2$ symmetric block seesaw mass matrix 
as given in (\ref{ssmm}), which with assumption (\ref{restss}) gives neutrino mass spectrum (\ref{seesaw1}) and (\ref{seesaw2}).

Kanaya \cite{doi:10.1143/PTP.64.2278} and independently Schechter and Valle \cite{Schechter:1981cv} showed that it is possible to block diagonalize the seesaw mass matrix, up to the terms of the order $M_{R}^{-1}M_{D}$. 
In this case the mixing matrix takes the following form
\begin{small}
\begin{equation}
\left(
\begin{array}{ll}
1 - \frac{1}{2}M_{D}^{\dag}(M_{R}M_{R}^{\dag})^{-1}M_{D} & (M_{R}^{-1}M_{D})^{\dag} \\
-M_{R}^{-1}M_{D} & 1 -\frac{1}{2}M_{R}^{-1}M_{D}M_{D}^{\dag}(M_{R}^{\dag})^{-1}
\end{array}
\right).
\label{Wappr}
\end{equation}
\end{small}

The seesaw mechanism can be neatly connected with the effective theory 
\cite{Weinberg:1979sa} in which
 \begin{equation}\label{eq: 2.1.7}
 \mathcal{L}_{eff}=- \frac{1}{\Lambda} \sum_{l',l} y_{l'l}(\Psi_{l'L}^{T}\sigma_{2} \Phi)\mathcal{C}^{-1}(\Phi^{T} \sigma_{2} \Psi_{lL}) + H.c.
 \end{equation}
 and
$ \psi_{lL}^T=
 \left(
 \nu_{lL},
 l_{L}
 \right), 
 \Phi^T=
 \left(
 \Phi^{+},
 \Phi^{0}
 \right)$
 are the SM lepton  and   Higgs doublets, respectively. Coefficients $y_{l' l}$ denotes dimensionless couplings and $\Lambda$ is the energy scale in which new physics effects do not decouple. After  spontaneous electroweak symmetry breaking, $\mathcal{L}_{eff}$ takes the form 
 \begin{equation}
\mathcal{L}_{eff} \to \mathcal{L}_{L}= -\frac{1}{2} \bar{\nu}_{L}\mathcal{M}_{L}(\nu_{L})^{\mathcal{C}} + H.c.
 \end{equation}
 with
 \begin{equation}
 \mathcal{M}_{L}=\frac{yv^{2}}{\Lambda}.
 \end{equation}
 There exist many ways to extend SM which results in the effective Lagrangian (\ref{eq: 2.1.7}). On the other hand, if we assume that we complete SM by adding only one type of particle, we are constrained to three possibilities to induce a light neutrino spectrum.
 These three possibilities lead to the different realization of the seesaw mechanism:
 \begin{enumerate}
 \item Seesaw Type-I (canonical seesaw).\\
 In this case, we add to SM right-handed neutrino fields $\nu_{R}$. Thus, we get a seesaw formula discussed previously 
 \begin{equation}
\mathcal{M}_{L}\to  M_{light} \simeq -M_{D}^{T}M_{R}^{-1}M_{D}, \quad \vert M_{D} \vert \ll \vert M_{R} \vert.
 \end{equation}
 
 Here $\Lambda$ in (\ref{eq: 2.1.7}) is identified with inverse of matrix $M_R$.
 \item Seesaw Type-II \cite{Magg:1980ut} \cite{Mohapatra:1980yp}. \\
 Instead of $\nu_{R}$ we can add a scalar boson triplet  $\Delta=(\Delta^{++}, \Delta^{+}, \Delta^{0})$ to get a small neutrino masses. Here $\Lambda$ in (\ref{eq: 2.1.7}) is identified with masses of scalar boson triplets. In this variation of the seesaw model, the mass of neutrinos is given by
 \begin{equation}
 \mathcal{M}_{L}\to  M_{light} \simeq \frac{\mu v^{2}}{M_{\Delta}^{2}}, \quad \vert \mu \vert \sim \vert M_{\Delta} \vert, \ \vert v \vert \ll \vert M_{\Delta} \vert,
 \end{equation}
 where $M_{\Delta}$ corresponds to the mass of the boson triplet $\Delta$.
 
 \item Seesaw Type-III \cite{Foot:1988aq}.\\
 In the last case we complement SM with a fermion triplet $\Sigma=( \Sigma^{+}, \Sigma^{0}, \Sigma^{-})$ 
which corresponds to $\Lambda$ in (\ref{eq: 2.1.7}). 
 In Type-III mechanism we get the following formula for neutrinos masses
 \begin{equation}
 \mathcal{M}_{L}\to  M_{light} \simeq -y^{T}M_{\Sigma}^{-1}y v^{2}, \quad \vert y \vert \ll \vert M_{\Sigma} \vert,
 \end{equation}
 where $M_{\Sigma}$ corresponds to the mass of the fermion triplet $\Sigma$.
 
 \end{enumerate}
 
 \subsection{Extended seesaw mass matrices}
 
Now we focus on seesaw extensions connected with extra fermion fields, which is a wide area of studies. In most of them, besides the right-handed neutrino fields $\nu_{R}$ characteristic for the seesaw type-I model, new singlet fermion fields $S_{R}$ are added. This type of the extension of the seesaw mechanism was introduced for the first time in 1983 by Wyler and Wolfenstein \cite{Wyler:1982dd}. The corresponding general mass term takes the following form
\begin{equation}
\begin{split}
\mathcal{L}_{ESS}= - \bar{\nu}_{L}M_{D}\nu_{R} - \bar{\nu}_{L}M_{L}(\nu_{L})^{\mathcal{C}} - \overline{(\nu_{R})^{\mathcal{C}}}M_{R}\nu_{R} \\
- \bar{\nu}_{L}M S_{R} - \overline{(\nu_{R})^{\mathcal{C}}}mS_{R} - \overline{(S_{R})^{\mathcal{C}}}\mu S_{R} + h.c.
\end{split}
\end{equation}
where $M_{D}$, $M, m$  are matrices of the Dirac type and $M_{L}$, $M_{R}$, $\mu$ are matrices of the Majorana type.
This mass term can be written in a compact form, in a similar way as it was made for an ordinary seesaw 
\begin{equation}
\mathcal{L}_{ESS}= -\bar{N}_{L}M_{ESS}N_{L}^{\mathcal{C}} + h.c.
\end{equation}
with the symmetric mass matrix $M_{ESS}$ 
\begin{equation}
M_{ESS}=
\left(
\begin{array}{ccc}
M_{L} & M_{D} & M \\
M_{D}^{T} & M_{R} & m \\
M^{T} & m^{T} & \mu 
\end{array}
\right).
\end{equation}

However, the most popular extensions of the seesaw mechanism with additional singlet fields, namely an inverse seesaw (ISS) and a linear seesaw (LSS), use a less general structure of $M_{ESS}$
\begin{small}
\begin{equation}
\begin{split}
&M_{ISS}=
\left(
\begin{array}{ccc}
0 & M_{D} & 0 \\
M_{D}^{T} & 0 & m \\
0 & m^{T} & \mu 
\end{array}
\right), \quad
M_{LSS}=
\left(
\begin{array}{ccc}
0 & M_{D} & M \\
M_{D}^{T} & 0 & m \\
M^{T} & m^{T} & 0 
\end{array}
\right) \\
&\begin{rm} where: \ \ \ \end{rm} \ \mu \ll M_{D} \ll m, \quad \quad \quad \quad \quad \quad \ \  
\ M_{D} \sim M \ll m. 
\end{split}
\end{equation}
\end{small}

It is important that for both linear and inverse seesaw mechanisms, we can rearrange mass matrices in a way that they will have the same structure as \eqref{ssmm}:
\begin{eqnarray}
&&ISS: \nonumber
\\
&&\mathcal{M}_{D} =
\left(
M_{D}, 
0
\right), \quad
\mathcal{M}_{R}=
\left(
\begin{array}{cc}
0 & m \\
m^{T} & \mu 
\end{array}
\right),\\
&&\nonumber \\
&&LSS: \nonumber
\\
&&\mathcal{M}_{D} =
\left(
M_{D}, 
M
\right), \quad
\mathcal{M}_{R}=
\left(
\begin{array}{cc}
0 & m \\
m^{T} & 0 
\end{array}
\right).
\label{ISSLSS}
\end{eqnarray}
Due to this rearrangement, these models can be analysed in the same way as the canonical seesaw \eqref{ssmm}, with the same hierarchy of elements, i.e. $ \vert \mathcal{M}_{D} \vert \ll \vert \mathcal{M}_{R} \vert$.

Using relations \eqref{congr} and \eqref{Wappr}, we get the following formula for the light neutrino sector  in the $ISS$ case
\begin{equation}
M_{light}=M_{D}m^{-1}\mu(m^{-1})^{T}M_{D}^{T}.
\end{equation}
In the inverse seesaw scenario, a small neutrino mass is obtained by double suppression of the Dirac mass matrix. Firstly, it is suppressed by the matrix $\mu$ and the second source of suppression lies in the inverse of the matrix $m$, which has elements much larger than those of $M_{D}$. It implies that the order of magnitude of elements of $m$ can be smaller than corresponding elements of $M_{R}$ in the canonical seesaw. Thus, it is more plausible to detect such heavy neutrino states in high energy colliders. 

Similarly, in the $LSS$ case the light neutrinos sector is given by the succeeding formula
\begin{equation}
M_{light}=  -  M_{D}m^{-1}M - M^{T}(m^{-1})^{T}M_{D}^{T} .
\end{equation}
Here, the light neutrino sector depends linearly on the Dirac mass matrix $M_{D}$, in contrast to quadratic dependence in the ordinary seesaw mechanism. 

We can see that despite differences in the structure of linear, inverse and type-I seesaw scenarios, corresponding mass matrices can be expressed uniformly by one general matrix
\begin{equation}
\label{general_mass}
\mathcal{M}=
\left(
\begin{array}{cc}
0 & \mathcal{M}_{D} \\
\mathcal{M}_{D}^{T} & \mathcal{M}_{R}
\end{array}
\right),
\end{equation}
with
\begin{equation}
\label{hierarchy}
\vert \mathcal{M}_{D} \vert \ll \vert \mathcal{M}_{R} \vert.
\end{equation}

 The structure of the seesaw mass matrix \eqref{general_mass} with the assumption \eqref{hierarchy} is the starting point for analysis of the general properties of arising eigenvalues and eigenvectors.   

\section{The structure of the seesaw mass matrix and neutrino masses}

The matrix $\mathcal{M}$ in \eqref{general_mass} can be split into the sum of two matrices with the different scale of elements
\begin{equation}
\label{eq:split}
\begin{split}
\mathcal{M}&=
\left(
\begin{array}{cc}
0 & \mathcal{M}_{D} \\
\mathcal{M}_{D}^{T} & \mathcal{M}_{R}
\end{array}
\right)=
\left(
\begin{array}{cc}
0 & 0 \\
0 & \mathcal{M}_{R}
\end{array}
\right)+
\left(
\begin{array}{cc}
0 & \mathcal{M}_{D} \\
\mathcal{M}_{D}^{T} & 0
\end{array}
\right)\\
&\equiv
\hat{\mathcal{M}}_{R}+\hat{\mathcal{M}}_{D}
\end{split}
\end{equation}
Such a split gives us an opportunity to use the theorem from the matrix analysis \eqref{eq:B.2.1} (eigenvalues and singular values of a given matrix M will be denoted by $\lambda_{i}(M)$ and $\sigma_{i}(M)$ $i=1,2,\dots,n$, respectively, see Appendix A for relevant definitions) which connects the spectrum of the sum of two matrices with the spectrum of those matrices. However, it is true only for Hermitian matrices. Therefore, at the beginning we will consider real symmetric mass matrix which implies the conservation of the CP symmetry, see e.g. \cite{Kayser:1984ge,Bilenky:1987ty,Gluza:1995ix,delAguila:1996ex,Giunti:2007ry,Gluza:2016qqv}. In what follows,
the CP invariance will be identified with a real symmetric mass matrix. With these assumptions, we get the following result.
\begin{prop}
\label{prop1}
In the CP invariant seesaw scenario with 
\begin{small}$\lambda(\mathcal{M}_{R}) \gg \vert \mathcal{M}_{D} \vert $, $\mathcal{M}_{D} \in M_{3 \times n}, \ \mathcal{M}_{R} \in M_{n \times n}$\end{small}, exactly 3 light neutrinos are present.
\end{prop}
\begin{proof}
In the seesaw model, we assume 
two well-separated scales of elements of the mass matrix \eqref{hierarchy}.
Let us split the mass matrix $\mathcal{M}$ according to these scales into the sum \eqref{eq:split}.
Weyl's inequalities (\ref{eq:B.2.1}) can be transformed into the following inequality
\begin{equation}\label{eq:4.0.4}
\vert \lambda_{i}(\mathcal{M})- \lambda_{i}(\hat{\mathcal{M}}_{R}) \vert \leq \rho(\hat{\mathcal{M}}_{D}).
\end{equation}
The spectral radius is smaller than each matrix norm \eqref{th:B.3.1}, in particular 
\begin{equation}
\begin{split}
\rho(\hat{\mathcal{M}}_{D}) \leq \Vert \hat{\mathcal{M}}_{D} \Vert_{2} = \Vert \mathcal{M}_{D} \Vert_{2} \leq \Vert \mathcal{M}_{D} \Vert_{F}, 
\end{split}
\end{equation}
where $\Vert \ast \Vert_{2}$ is an operator norm. We used the fact that the operator norms of $\hat{\mathcal{M}}_{D}$ and $\mathcal{M}_{D}$ are equal and also the following relation between the operator and Forbenius norm $\Vert \ast \Vert_{2} \leq \Vert \ast \Vert_{F}$.
Since all elements of $\mathcal{M}_{D}$ have the same order of magnitude, the following estimation can be made
\begin{equation}
\begin{split}
&\Vert \mathcal{M}_{D} \Vert_{F}=\sqrt{\sum_{i,j}  \vert (\mathcal{M}_{D})_{ij} \vert^{2}} \leq \sqrt{3n}\vert \mathcal{M}_{D} \vert \Rightarrow \\
&\Rightarrow \rho(\hat{\mathcal{M}}_{D})\leq \sqrt{3n}\vert \mathcal{M}_{D} \vert.
\end{split}
\end{equation}
On the other hand, we know that matrix $\hat{\mathcal{M}}_{R}$ has at least three eigenvalues equal to 0. Since the eigenvalues of the Hermitian matrix can be arranged as in \eqref{eq:B.1.4}, we have 
\begin{equation}
\vert \lambda_{i}(\mathcal{M})- 0 \vert \leq \rho(\hat{\mathcal{M}}_{D}) \quad  for  \ \lambda_{i}(\hat{\mathcal{M}}_{R})=0.
\end{equation}
Hence, three eigenvalues of $\mathcal{M}$ must be smaller than $\sqrt{3n} \vert \mathcal{M}_{D} \vert $. \textbf{This means that at least three light neutrinos exist}.  \\
Now, let us assume that $\vert \lambda(\mathcal{M}_{R}) \vert \gg \vert \mathcal{M}_{D} \vert$. With this assumption we can show that all remaining eigenvalues of $\mathcal{M}$ must be large, i.e. of the order of elements of $\mathcal{M}_{R}$.
This is another conclusion from Weyl's inequalities.
(\ref{eq:4.0.4}) tells us that {\bf eigenvalues of $\mathcal{M}$ are maximally shifted by $\vert \mathcal{M}_{D} \vert $ from eigenvalues of $\hat{\mathcal{M}}_{R}$}. However, all eigenvalues of $\hat{\mathcal{M}}_{R}$ which remained are equal to eigenvalues of $\mathcal{M}_{R}$.  Therefore,
\begin{equation}
\vert \lambda(\mathcal{M}_{R}) \vert \gg \vert \mathcal{M}_{D} \vert \Rightarrow \lambda(\mathcal{M})\simeq \lambda(\mathcal{M}_{R}) \pm \vert \mathcal{M}_{D} \vert . 
\label{splitting}
\end{equation}
\textbf{Thus, in the CP invariant seesaw scenario with $\vert \lambda(\mathcal{M}_{R}) \vert \gg \vert \mathcal{M}_{D} \vert$, exactly three light neutrinos exist, and all additional states must be heavy}.
\end{proof}

As the matrix $\mathcal{M}_D$ couples left-handed and right-handed chiral fermions, (\ref{splitting}) gives an estimate on the mass of heavy neutrinos to be the mass obtained by neglecting light neutrinos shifted at most by the maximal strength of the coupling to the light neutrino sector.

The above discussion relies on eigenvalues of mass matrices.
However, eigenvalues are not good quantities for general neutrino mass scenarios with complex symmetric matrices. Such matrices can have complex eigenvalues, and moreover, they are not always diagonalizable by the unitary similarity transformation. 
Instead, singular values are useful, since due to Autonne-Takagi theorem \eqref{th:B.1.4}, we can always find a unitary transformation which will diagonalize the complex seesaw mass matrix. Thus, the above result can be generalized to the complex seesaw scenario with the usage of the analog of Weyl's inequalities for singular values.
\begin{cor}
\label{cor1}
In the seesaw scenario with 
\begin{small}$\sigma(\mathcal{M}_{R}) \gg \vert \mathcal{M}_{D} \vert $, $\mathcal{M}_{D} \in M_{3 \times n}, \ \mathcal{M}_{R} \in M_{n \times n}$\end{small}, exactly 3 light neutrinos are present.
\end{cor}
\begin{proof}
For singular values, we have the analog of Weyl's inequalities
\begin{equation}
\vert \sigma_{i}(\mathcal{M}) -\sigma_{i}(\hat{\mathcal{M}}_{R})\vert \leq \sigma_{1}(\hat{\mathcal{M}}_{D}) = \Vert \hat{\mathcal{M}}_{D} \Vert_{2},
\end{equation}
\\
thus using similar arguments as it was done in the proof of Proposition \eqref{prop1}, we attain the assertion.
\end{proof}

One of the main seesaw mechanism assumptions is that beside three light neutrinos all additional should be very massive. The masses of heavy neutrinos are dominated by the spectrum of the $\mathcal{M}_{R}$ submatrix. However, as we could see from the simple example \eqref{eq:example_1} even when the elements of the matrix are of the same order, eigenvalues can fall apart. As a consequence, some eigenvalues of $M_{heavy}$ could be very small, which contradicts the seesaw assumptions. Such possibility can be easily seen from the Laguerre–Samuelson inequality for the real roots of polynomials. 
Our goal is to establish conditions under which eigenvalues of the $\mathcal{M}_{R}$ matrix are always large. In the case of one additional neutrino, the situation is trivial since $\mathcal{M}_{R}$ is represented just by one number. Moreover, in this case, the seesaw mechanism is no longer valid, since the light spectrum contains two massless neutrinos, which contradicts experimental results.
With two and more additional neutrinos the general solution to the problem is very difficult. In principle, we want to exclude the region smaller than some boundary value. In general, the $\mathcal{M}$ matrix is complex symmetric and masses correspond to the singular values.
Thus, if we are able to find a lower bound for the smallest singular value and impose it to be larger than some limit value, the problem is solved. Moreover, in the CP invariant scenario, $\mathcal{M}$ can be considered as real and symmetric, in which case masses are given by eigenvalues. It is known that for normal matrices singular values are equal to the absolute values of the eigenvalues \cite{horn_johnson_1991}. Thus, by bounding singular values we treat both cases.
In the mathematical literature, there are available different lower bounds for the smallest singular value expressed by matrix elements \cite{HONG199227,YISHENG199725,JOHNSON1998169,ROJO1999215,PIAZZA2002141,zou_2012}.  
We follow \cite{YISHENG199725,PIAZZA2002141} where
\begin{equation}
\sigma_{n}(A) \geq \vert detA \vert \left( \frac{n-1}{\Vert A \Vert_{F}^{2}}  \right)^{\frac{n-1}{2}} \geq X,
\end{equation}
and $n$ is the dimension of the matrix. $X$ is some boundary value and $A$ is a matrix with elements $a_{11},a_{12},...,a_{nn}$. This nontrivial inequality can be solved analytically in two dimensions which corresponds to the scenario with two additional neutrinos, using e.g. \texttt{Mathematica} \cite{Mathematica}. This scenario known as the minimal seesaw is currently intensively studied \cite{Frampton:2002qc, Guo:2006qa,Zhang:2009ac,Yang:2011fh, Antusch:2015mia,Chianese:2018dsz,Chianese:2019epo}. The analytic formulas restrict matrix elements to singular values larger than some positive number $X$, which in case of eigenvalues correspond to the region outside the interval $(-X,X)$.  
Using the abbreviations:
\begin{widetext}
\makeatletter
    \def\tagform@#1{\maketag@@@{\normalsize(#1)\@@italiccorr}}
\makeatother
\begin{eqnarray}
Y_1&=&\frac{a_{11} a_{12}^2}{a_{11}^2-X^2}, \hspace*{4cm} Y_2=\sqrt{\frac{a_{11}^4 X^2+2 a_{11}^2 a_{12}^2 X^2-a_{11}^2 X^4+a_{12}^4 X^2-2 a_{12}^2 X^4}{\left(a_{11}^2-X^2\right)^2}},\\ Y_3&=&\frac{-a_{11}^2 X^2+a_{12}^4-2 a_{12}^2 X^2}{2 a_{11} a_{12}^2},
\hspace*{1.7cm} Y_4=\sqrt{\sqrt{X^4-a_{11}^2 X^2}-a_{11}^2+X^2},
\end{eqnarray}

we get:

\begin{equation}
\label{heavy_form}
\begin{split}
a_{12},a_{22} \in \mathbb{R} \land & \Bigg\{
\bigg( a_{11}>X\land X>0\land a_{22}\geq Y_1+Y_2\bigg)\lor 
\bigg( a_{11}>X\land X>0\land a_{22}\leq Y_1-Y_2\bigg)\lor \\
&\bigg( X>0\land a_{22}\geq Y_1+Y_2\land a_{11}<-X\bigg)\lor 
\bigg(X>0\land a_{11}<-X\land a_{22}\leq Y_1-Y_2\bigg)\lor \\
&\bigg( -X=a_{11}\land a_{12}>0\land X>0\land a_{22}\geq  Y_3 \bigg)\lor 
\bigg( -X=a_{11}\land X>0\land a_{22}\geq Y_3 \land a_{12}<0\bigg)\lor \\
&\bigg( X=a_{11}\land a_{12}>0\land X>0\land a_{22}\leq  Y_3 \bigg)\lor 
\bigg( X=a_{11}\land X>0\land a_{12}<0\land a_{22}\leq  Y_3 \bigg)\lor \\
&\bigg(-Y_4=a_{12}\land Y_1-Y_2=a_{22}\land X>0\land a_{11}<X\land -X<a_{11} \bigg)\lor \\ 
&\bigg( Y_4=a_{12}\land Y_1-Y_2=a_{22}\land X>0\land a_{11}<X\land -X<a_{11} \bigg)\lor \\
&\bigg( a_{12}>Y_4\land X>0\land a_{11}<X\land -X<a_{11}\land a_{22}\leq Y_1+Y_2\land Y_1-Y_2\leq a_{22}\bigg)\lor \\
&\bigg( X>0\land a_{11}<X\land a_{12}<-Y_4\land -X<a_{11}\land a_{22}\leq Y_1+Y_2\land Y_1-Y_2\leq a_{22}\bigg)
\bigg\},
\end{split}
\end{equation}
\end{widetext}
To approach higher-dimensional cases we use the inverse eigenvalue problem, i.e. reconstruction of matrices from the spectrum \cite{Chu_1998, chu_golub_2002}. Such reconstruction for Hermitian matrices, i.e. in CP invariant scenario, is controlled by the Schur-Horn theorem \cite{schur_1923,horn_1954} 

\begin{theorem}{(Schur-Horn)}
Let $\lbrace \lambda_{i} \rbrace_{i=1}^{n}$ and $\lbrace d_{i} \rbrace_{i=1}^{n}$ 
be vectors in $\mathbb{R}^{n}$ with entries in non-increasing order. There is a Hermitian matrix with diagonal entries $\lbrace d_{i} \rbrace_{i=1}^{n}$ and eigenvalues $\lbrace \lambda_{i} \rbrace_{i=1}^{n}$ if and only if
\begin{equation}\label{eq:major}
\begin{split}
&\sum_{i=1}^{k} d_{i} \leq \sum_{i=1}^{k} \lambda_{i} \quad k=1,...,n \\
&\text{and} \\
&\sum_{i=1}^{n} d_{i} = \sum_{i=1}^{n} \lambda_{i}.
\end{split}
\end{equation}
\end{theorem}

The construction of matrices based on this theorem can be realized by different approaches \cite{FRIEDLAND1979412,CHAN1983562,chu_1995, Zha1995,fickus_2011}.
An important feature of this theorem is the majorization condition between eigenvalues and diagonal elements \eqref{eq:major}. In the case of large eigenvalues, this relation implies that diagonal elements also must be large in comparison to off-diagonal elements. However, this works only if the $\mathcal{M}_{R}$ matrix is non-negative definite, i.e. all eigenvalues are non-negative. Such a situation in the case of three sterile neutrinos is presented in (Fig.~\ref{fig1}). In a scenario with some eigenvalues large but negative, the relation \eqref{eq:major} does not restrict matrix elements.  

To overpass the requirement of non-negative definiteness and CP conservation we can invoke singular values once again. As in the eigenvalue case, singular values can also be used to reconstruct a matrix via a procedure known as inverse singular value problem \cite{Chu:2000,SHENGJHIH2013}. However, currently on our disposal we have only theorems which connect eigenvalues and singular values (Weyl-Horn theorem \cite{Weyl408,Horn1954OnTE}) or singular values and diagonal elements (Sing-Thompson theorem \cite{sing_1976,thompson_1977}). Thus, we miss symmetry of the matrix and further work is needed to combine all these components.

\begin{figure}[h!]
\begin{center}
\includegraphics[scale=0.4]{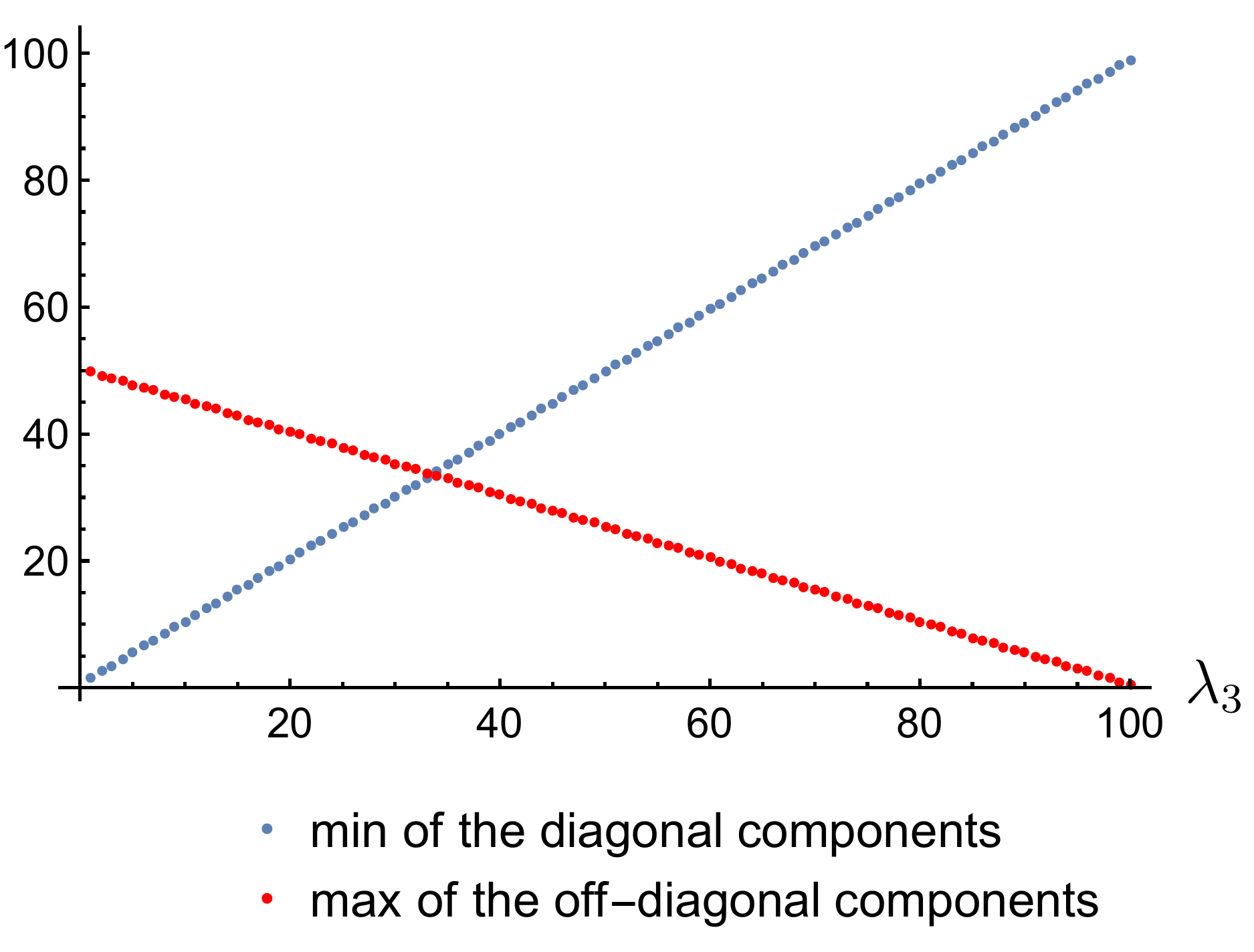}
\end{center}
\caption{An illustration of the Schur-Horn theorem for the non-negative definite matrix. Two eigenvalues have been set up to $\lambda_{1}=101$ and $\lambda_{2}=100$, and the third eigenvalue $\lambda_{3}$ ranges from 0 to 100 to see the behavior of the diagonal elements and off-diagonal elements for a different spread between eigenvalues. When the spread is large, i.e. $\lambda_{3} \sim 0$, the diagonal elements can be very small and off-diagonal elements can take significant values. On the other hand, if all eigenvalues ale large then diagonal elements dominate.}
\label{fig1}
\end{figure}

\section{The separation between eigenspaces in the seesaw scenario}

We are interested how masses and mixings are connected to each other in the seesaw scenario\footnote
{Recently, an interesting relation has been found between eigenvectors and eigenvalues for neutrino oscillations   in \cite{Denton:2019ovn, Denton:2019pka}.}. To answer this we will study the behavior of the eigenspace of the matrix $\hat{\mathcal{M}}_{R}$ under the perturbation $\hat{\mathcal{M}}_{D}$ \eqref{eq:split}. Thus, we are interested in the estimation of the difference between eigenspaces spanned by eigenvectors of $\hat{\mathcal{M}}_{R}$ and $\mathcal{M}$ \eqref{eq:split}. 
As a starting point let us consider the eigenproblem for the matrix $\hat{\mathcal{M}}_{R}$. For block-diagonal matrices eigenvalues corresponds to the eigenvalues of its diagonal blocks. In this case, one of these blocks is a zero matrix. Thus, this block has a threefold eigenvalue 0 and the corresponding eigenvectors are
\begin{eqnarray}
&& (1,0,0,0,0,0,...,0)^{T},
(0,1,0,0,0,0,...,0)^{T}, \nonumber \\
&& (0,0,1,0,0,0,...,0)^{T}.
\end{eqnarray} 
They span a standard 3-dimensional Euclidean space embedded in a (3+n)-dimensional space. The rest of the eigenvalues of $\hat{\mathcal{M}}_{R}$ correspond to those of the $\mathcal{M}_{R}$ submatrix.
Our approach will be based on the Davis-Kahan theorem \eqref{DK} which is valid for the CP-conserving case (a generalization to the CP-violating case seems to be possible \cite{IPSEN2000131}, however it requires a separate study). It allows us to estimate the sine of the angle between subspaces, denoted as $\sin\Theta$, spanned by the eigenvectors. Since the eigenspace spanned by the zero eigenvalues of $\hat{\mathcal{M}}_{R}$ has a very simple structure we will focus on the estimation of the angle between spaces corresponding to light neutrinos. Information about the other pair of subspaces follows immediately from the orthogonality of the mixing matrix. Let us denote the eigenspaces spanned by the eigenvectors corresponding to small eigenvalues by $V_{L}$ and $V_{L}^{'}$, respectively for $\hat{\mathcal{M}}_{R}$ and $\mathcal{M}$. Then in the seesaw scenario (see Fig.~\ref{fig:DK}) we have 
\begin{equation}\label{7.2}
\Vert \sin\Theta(V_{L},V_{L}^{'}) \Vert \leq \frac{1}{\delta} \Vert \mathcal{M} - \hat{\mathcal{M}}_{R} \Vert =\frac{1}{\delta} \Vert \mathcal{M}_{D} \Vert,
\end{equation}
where $\delta$ is the distance between the largest of the light masses and the smallest of the heavy masses.
The above inequality says that $\sin\Theta(V_{L},V_{L}^{'})$ can be estimated using a gap between spectra and the size of the perturbation. It is clear that if the subspaces $V_{L}$ and $V_{L}^{'}$ are close to each other then the sine between them will tend to zero.
Therefore, from \eqref{7.2} we can draw the following conclusions:
\begin{itemize}
\item If the separation between light and heavy neutrinos is pronounced like in the seesaw case, then the subspace spanned by light neutrinos is almost parallel to the 3-dimensional Euclidean space. However, when these two spectra approach each other not much information can be retrieved from \eqref{DK}. 
\item Even if the $\delta$ is not that large, these two subspaces still can be almost parallel when $\mathcal{M}_{D}$ is very small.
\end{itemize}

\begin{figure}[h!]
\includegraphics[scale=0.4]{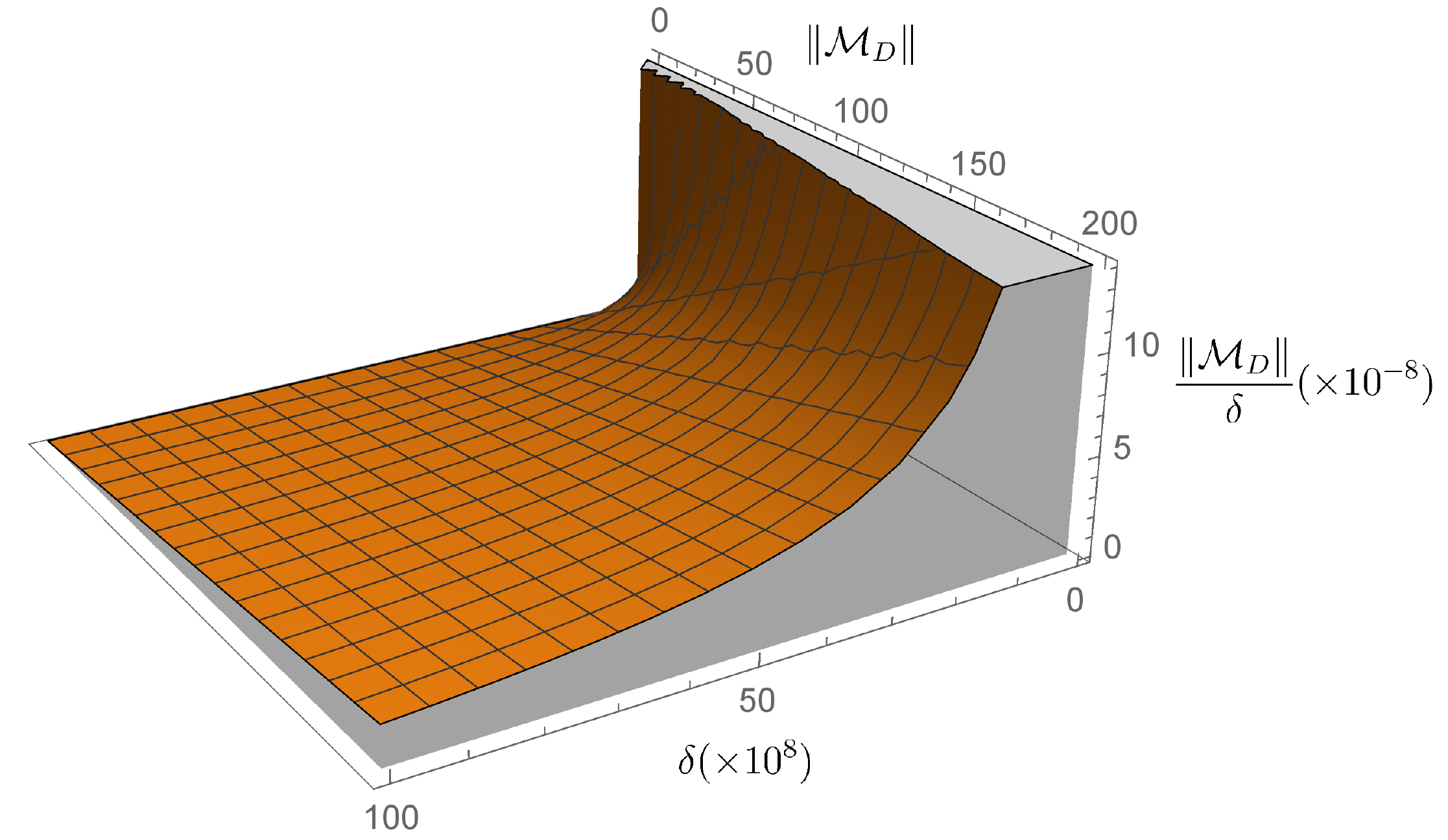}
\caption{The behavior of the $\sin \Theta$, controlled by the Davis-Kahan theorem, in the seesaw scenario. The region below the graph represents allowed values. In this case, $\sin \Theta$ is bounded by a function depending on the norm of $\mathcal{M}_{D}$ and the gap between the spectrum. As heavy neutrinos become lighter and lighter the blowout of the bound is observed.}
\label{fig:DK}
\end{figure}

\section{Summary and outlook}

Simple ideas are very often the most powerful, and this is the case for the seesaw mechanism which provides an attractive way to explain the smallness of the light neutrino masses by introducing very massive sterile neutrino states. This manifests in a specific structure of the mass and mixing matrices. We treat uniformly various seesaw types of mass matrices, including linear and inverse extensions, using the same, general and rigid-block mass matrix structure. We proved that under the general sub-matrix mass hierarchies \eqref{hierarchy} exactly three light neutrinos emerge (Prop.~\ref{prop1} and Cor.~\ref{cor1}). Moreover, as a consequence we derived the allowed splitting for heavy neutrinos in terms of submatrices $\mathcal{M}_{D}$ and $\mathcal{M}_{R}$ \eqref{splitting}. As the spectrum of $\mathcal{M}_{R}$ dominates the contribution to heavy masses, we have investigated the structure of this matrix to ensure spectrum large. In a minimal seesaw scenario with two sterile neutrinos, we gave analytic bounds for heavy neutrino masses expressed by the matrix elements \eqref{heavy_form}. For cases with a larger number of additional neutrinos the inverse eigenvalue problem has been applied, however this can be done systematically only in CP invariant case and for positive definite matrices. The general solution for any dimension still requires more study. Especially, the inverse singular value methods could be useful. This requires connection of currently available theorems with the specific structure of the seesaw mass matrix. 
Lastly, we studied the behavior of the angle between subspaces spanned by the eigenvectors which connects masses with mixings. In this case the Davis-Kahan theorem applied to the seesaw mechanism gives a simple estimation of the angle between mixing spaces depending on the norm of the Dirac mass matrix. 

Our work is based on matrix theory which is a vast and rich field. We would like to outline a few potential directions related to neutrino physics for further studies:
\begin{itemize}
\item Gershgorin circles provide alternative inclusive entrywise bounds for eigenvalues. It can be applied to models with diagonally dominant mass matrix to get insight into the mass spectrum.
\item Symmetric gauge functions are strictly connected to the unitary invariant norms.  We use unitary invariant norms in our study of the mixing matrices \cite{Bielas:2017lok,Flieger:2019nsb,Flieger:2019eor}. The symmetric gauge functions can provide a new perspective into the mixing analysis. 
\item The characteristic polynomial with real roots discussed in this work is a particular example of hyperbolic polynomials. This gives the opportunity to study eigenvalue problems from a more general point of view.
\item Semidefinite programming (SDP) does not come directly from matrix theory. However, this part of the mathematical programming is based on the positive-definite matrices. SDP can be used to better understand the region of physically admissible mixing matrices \cite{Bielas:2017lok, Saunderson_2015}.
\end{itemize}

\section*{Acknowlegments}
We would like to thank Krzysztof Bielas and Marek Gluza for useful remarks. The work was supported partly by the Polish National Science Centre (NCN) under the Grant Agreement 2017/25/B/ST2/01987 and the COST (European Cooperation in Science and Technology) Action CA16201 PARTICLEFACE.

\section*{Appendix: \\Matrix theory insight to the seesaw mass matrix studies}

In this appendix we introduce definitions and theorems used in the main text. Proofs for presented here statements can be found in \cite{horn_johnson_2012, Meyer:2000:MAA:343374, rao}.

\subsection{Matrix norms}

Let us begin with consideration the matrix "size" problem. 
A set of all matrices of a given dimension along with matrix addition and matrix multiplication creates a vector space. Thus, it is natural to consider a size of vectors or a distance between two points of this space. This can be done by introducing a function called the norm.
\begin{defi}
A norm for a real or complex vector space $V$ is a function $\Vert \cdot \Vert$ mapping $V$ into $\mathbb{R}$ that satisfies the following conditions
\begin{equation}\label{vnorm}
\begin{split}
&\Vert A \Vert \geq 0 \ and \ \Vert A \Vert=0 \Leftrightarrow A=0, \\
&\Vert \alpha A \Vert = \vert \alpha \vert \Vert A \Vert, \\
&\Vert A+B \Vert \leq \Vert A \Vert + \Vert B \Vert. \\
\end{split}
\end{equation}
\end{defi}
The same is true for the matrix space, however, for matrices, this can be done in two ways. We can use either standard vector norm \eqref{vnorm} or introduce more adequate so-called matrix norm which takes into account specific matrix multiplication.  
\begin{defi}\label{mnorm}
A matrix norm is a function $\Vert \cdot \Vert$ from the set of all complex matrices into $\mathbb{R}$ that satisfies the following properties
\begin{equation}
\begin{split}
&\Vert A \Vert \geq 0 \ and \ \Vert A \Vert=0 \Leftrightarrow A=0, \\
&\Vert \alpha A \Vert = \vert \alpha \vert \Vert A \Vert, \\
&\Vert A+B \Vert \leq \Vert A \Vert + \Vert B \Vert, \\
&\Vert A B \Vert \leq \Vert A \Vert \Vert B \Vert.
\end{split}
\end{equation}
\end{defi}  
It is important to emphasize that usual vector norms \eqref{vnorm} and matrix norms \eqref{mnorm} are strictly connected: Any vector norm can be translated into a matrix norm in the following way
\begin{equation}
\Vert A \Vert_{\star} = \max_{\Vert x \Vert_{\star} = 1}\Vert Ax \Vert_{\star},
\end{equation}
where $\Vert \cdot \Vert_{\star}$ stands for a corresponding vector norm. The matrix norm defined in this way ensure submultiplicativity condition and it is called the induced matrix norm. The most popular matrix norms are:
\begin{itemize}
\item Spectral norm: $\Vert A \Vert  = \max_{\Vert x \Vert_{2} = 1} \Vert Ax \Vert_{2}=\sigma_{1}(A)$.

\item Frobenius norm: $\Vert A \Vert_{F} =\sqrt{Tr(A^{\dag}A)}=\sqrt{\sum_{i,j=1}^{n}\vert a_{ij} \vert^{2}}=\sqrt{\sum_{i=1}^{n}\sigma_{i}^{2}(A)}$.

\item Maximum absolute column sum norm: $\Vert A \Vert_{1} = \max_{\Vert x \Vert_{1} = 1} \Vert Ax \Vert_{1}=\max_{j} \sum_{i}\vert a_{ij} \vert $.

\item Maximum absolute row sum norm: $\Vert A \Vert_{\infty} = \max_{\Vert x \Vert_{\infty} = 1} \Vert Ax \Vert_{\infty}=\max_{i} \sum_{j}\vert a_{ij} \vert $.
\end{itemize}

\subsection{Eigenvalues and singular values}

Neutrinos with definite masses are obtained through a unitary transformation which brings the mass matrix into diagonal form. In a general seesaw scenario where diagonalization is done by the congruence transformation \eqref{congr}, masses are given by singular values. However, if we restrict attention to the CP invariant case, diagonalization goes through the similarity transformation, and the quantities corresponding to neutrino masses are eigenvalues. We will present theorem concerning both of these quantities, starting with the notion of a spectral radius.

\begin{defi}
Let $A \in M_{n}$. The spectral radius of A is $\rho(A)= \max \lbrace \vert \lambda \vert : \ \lambda \in \sigma(A) \rbrace $.
\end{defi}
All matrix norms and spectral radius are connected by the following theorem. 
\begin{theorem}
Let A be an $n \times n$ matrix, then for any matrix norm $\Vert \cdot \Vert$ the following statement is true
\begin{equation}
\rho(A)\leq \Vert A \Vert
\label{th:B.3.1}
\end{equation}
\end{theorem}

\begin{theorem}\label{th:B.1.1}
Let $A \in M_{n}$ be Hermitian. Then the eigenvalues of A are real.
\end{theorem}
Using this theorem we can arrange the eigenvalues of a given Hermitian matrix $A \in M_{n}$, e.g, in a decreasing order
\begin{equation}\label{eq:B.1.4}
\lambda_{1} \geq ... \geq \lambda_{n},
\end{equation}
and this convention is used in this work.

\begin{theorem}{(Spectral theorem for Hermitian matrices)}\\
A matrix $A \in M_{n}$ is Hermitian if and only if there is a unitary $U \in M_{n}$ and diagonal $\Lambda \in M_{n}$ such that $A=U \Lambda U^{\dag}$.
\end{theorem}

There exist an equivalent decomposition theorem for singular values.

\begin{theorem}{(Singular value decomposition)}\label{th:B.1.3}\\
Let $A\in M_{m \times n}$ be given and let $q=\min \lbrace m,n \rbrace$. Then there is a matrix $\Sigma=(\sigma_{ij}) \in M_{m \times n}$ with $\sigma_{ij}=0$ for all $i \neq j$ and $\sigma_{11}\geq \sigma_{22} \geq ... \geq \sigma_{qq}$, and there are two unitary matrices $V \in M_{m \times m} $ and $U \in M_{n \times n}$ such that $A= V \Sigma U^{\dag} $.
\end{theorem}

Autonne and Takagi \cite{autonne, takagi} gave us a criterion based on singular values, which characterizes the class of symmetric matrices. 

\begin{theorem}{(Autonne-Takagi)}\label{th:B.1.4}\\
Let $A \in M_{n}$. Then $A=A^{T}$ if and only if there is a unitary matrix $U \in M_{n}$ and a nonnegative diagonal matrix $\Sigma$ such that $A= U \Sigma U^{T} $. The diagonal entries of $\Sigma$ are the singular values of $A$.
\end{theorem}

Since there are matrices for which both sets of eigenvalues and of singular values are well defined, the natural question arises, how are these quantities connected? The following theorem provides the basic relation between these numbers.

\begin{theorem}\label{th:B.1.5}
Let $A \in M_{n}$ have singular values $\sigma_{1}(A) \geq ... \geq \sigma_{n}(A) \geq 0$ and eigenvalues $\lbrace \lambda_{1}(A), ..., \lambda_{n}(A) \rbrace \in \mathbb{C}$ ordered so that $\vert \lambda_{1}(A) \vert \geq ... \geq \vert \lambda_{n}(A) \vert $. Then
\begin{equation}
\begin{split}
&\vert \lambda_{1}(A)...\lambda_{k}(A) \vert \leq \sigma_{1}(A)...\sigma_{k}(A) \ for \ k=1,...,n \\ 
&with \ equality \ for \ k=n.
\end{split}
\end{equation}
\end{theorem}

Using the above definitions and basic theorems a theorem which bounds eigenvalues of the sum of two matrices can be formulated. In a general case, we can say almost nothing about eigenvalues of the sum of matrices. However, for Hermitian matrices, the situation is more accessible and we have a set of helpful relations. We will present only the main result provided by Weyl \cite{Weyl1912}, however, it can be extended to more specific cases. 
\begin{theorem}{(Weyl's inequalities)}\label{eq:B.2.1}\\
Let A and B be $n \times n$ Hermitian matrices. Then
\begin{equation}
\begin{split}
\lambda_{j}(A+B) \leq \lambda_{i}(A) + \lambda_{j-i+1}(B) \ for \ i \leq j \\ 
\lambda_{j}(A+B) \geq \lambda_{i}(A) + \lambda_{j-i+n}(B) \ for \ i \geq j
\end{split}
\end{equation}
\end{theorem}

After some work, the above relations can be transformed to the following form
\begin{equation}
\vert \lambda_{j}(A+B)-\lambda_{j}(A)\vert \leq \rho(B).
\end{equation}

Despite the fact that Weyl's inequalities can be used to estimate eigenvalues of the sum without any restriction to scale of its summands, they give the best results if one of the matrices can be treated as a small additive perturbation of the second matrix which is a case of the seesaw mechanism.

As singular values are defined as square roots of Hermitian matrix $A^{\dag}A$ we should expect that similar result to Weyl's inequalities is also valid for singular values. However, due to their nonegative nature, we can only estimate the singular values of the sum from above.

\begin{theorem}{(Weyl's inequality for singular values)}\\
Let A and B be a $m \times n$ matrices and let $q=\min \lbrace m,n \rbrace $. Then
\begin{equation}
\begin{split}
\sigma_{j}(A+B) \leq \sigma_{i}(A) + \sigma_{j-i+1}(B) \ for \ i \leq j \\ 
\end{split}
\end{equation}
\end{theorem}

\subsection{Eigenspace}

The behavior of eigenvectors of a matrix $A$ under the perturbation is much more complicated than that of the eigenvalues. However, in the case of subspaces spanned by eigenvectors there are theorems allowing quantitative prediction of their perturbation. Estimation of the difference between perturbated and unperturbed eigenspaces can be done with the help of orthogonal projections as the following example shows.
Let $S$ be eigenspace of $A$ spanned by some of its eigenvectors and let $S^{\perp}$ be its orthogonal complement. Then $A$ can be decomposed as
\begin{equation}
A=E_{0}A_{0}E_{0}^{\dag}+E_{1}A_{1}E_{1}^{\dag}
\end{equation}
where $E_{0}$ is the orhonormal basis for $S$ and $E_{1}$ is the orhonormal basis for $S^{\perp}$.
Similarly, for $\hat{A}=A+E$ and eigenspace $\hat{S}$ we have
\begin{equation}
\hat{A}=F_{0}\Lambda_{0}F_{0}^{\dag}+F_{1}\Lambda_{1}F_{1}^{\dag}   
\end{equation}
We would like to know how well vectors in $\hat{S}$ approximate vectors in $S$. 
The orthogonal projectors onto $S$ and $\hat{S}$ are given by $E_{0}E_{0}^{\dag}$ and $F_{0}F_{0}^{\dag}$ respectively. Every vector $x$ in $S$ can be written as $x=E_{0}\alpha$ where $\alpha \in \mathbb{C}^{dimS}$ and its projection onto $\hat{S}$ is $\hat{x}=F_{0}F_{0}^{\dag}E_{0}\alpha$. Thus
\begin{align}
\Vert x-\hat{x} \Vert 
&= \Vert E_{0}\alpha - F_{0}F_{0}^{\dag}E_{0}\alpha \Vert = 
\Vert (I - F_{0}F_{0}^{\dag})E_{0} \alpha \Vert = \nonumber \\
&=\Vert F_{1}F_{1}^{\dag}E_{0}\alpha \Vert = \Vert  F_{1}^{\dag}E_{0}\alpha \Vert. 
\end{align}
Hence $F_{1}^{\dag}E_{0}$ tells us how close $\hat{x}$ is to $x$.

Before we move to the main perturbation theorem, let us state the auxiliary theorem which highlights geometric aspects of the relation between subspaces \cite{bhatia1}.

\begin{theorem}
Let $X_{1},Y_{1}$ be $n \times l$ matrices with orthonormal columns.Then there exist $l \times l$ unitary matrices $U_{1}$ and $V_{1}$, and an $n \times n$ unitary matrix $Q$, such that if $2l \leq n$, then
\begin{equation}
QX_{1}U_{1}=
\left(
\begin{array}{c}
I \\
0 \\
0
\end{array}
\right),
\end{equation}

\begin{equation}
QY_{1}V_{1}=
\left(
\begin{array}{c}
C \\
S \\
0
\end{array}
\right),
\end{equation}
where $C,S$ are diagonal matrices with diagonal entries $0 \leq c_{1} \leq ... \leq c_{l} \leq 1$ and $1 \geq s_{1} \geq s_{1} \geq ... \geq s_{l} \geq 0$, respectively, and $C^{2} + S^{2} =I$.
\end{theorem}

The relation between matrices $C$ and $S$ resembles the relation between trigonometric functions. This allows us to define angles between subspaces.

\begin{defi}
Let $\mathcal{E}$ and $\mathcal{F}$ let be $l-$dimensional subspaces of $\mathbb{C}^{n}$. The angel operator between $\mathcal{E}$ and $\mathcal{F}$ is defined as follows
\begin{equation}
\Theta(\mathcal{E}, \mathcal{F})=\arcsin S.    
\end{equation}
It is a diagonal matrix whose diagonal elements are called the canonical (principal) angles between subspaces $\mathcal{E}$ and $\mathcal{F}$.
\end{defi}

Moreover, using the matrix norm we can define the gap between two subspaces.

\begin{defi}
Let $\mathcal{E}$ and $\mathcal{F}$ let be $l-$dimensional subspaces of $\mathbb{C}^{n}$. Let $E$ and and $F$ be orthogonal projection onto $\mathcal{E}$ and $\mathcal{F}$ respectively. The distance between subspaces $\mathcal{E}$ and $\mathcal{F}$ is defined to be
\begin{equation}
\Vert E-F \Vert = \Vert E^{\perp}F \Vert = \Vert \sin \Theta \Vert  
\end{equation}
\end{defi}

The perturbation behavior between eigenspaces of Hermitian matrices is described by the renown Davis-Kahan theorem \cite{Davis}.
 
\begin{theorem}\label{DK}
Let $A$ and $B$ be Hermitian operators, and let $S_{1}$ be an interval $[a,b]$ and $S_{2}$ be the complement of $(a - \delta, b + \delta)$ in $\mathbb{R}$. Let $E=P_{A}(S_{1}), F^{\perp}=P_{B}(S_{2})$ be orthogonal projections onto subspaces spanned by eigenvectors of $A$ and $B$ corresponding to eigenvalues from $S_{1}$ and $S_{2}$ respectively. Then for every unitarily invariant norm, 
\begin{small}
\begin{equation}
||| EF^{\perp}||| \leq \frac{1}{\delta} |||E(A-B)F^{\perp}||| \leq \frac{1}{\delta} |||A-B|||,
\end{equation}
where 
\begin{equation}
\delta=dist(\sigma(A),\sigma(B))= min \lbrace | \lambda - \mu |: \lambda \in \sigma(A), \mu \in \sigma(B) \rbrace.
\end{equation}
\end{small}
\end{theorem}

\providecommand{\href}[2]{#2}
\bibliographystyle{elsarticle-num}
\bibliography{bibliografia}

\end{document}